\documentclass{article}

\usepackage[english]{babel}
\usepackage{amsmath}
\usepackage{amsfonts}
\usepackage{amssymb}
\usepackage{amsthm}
\usepackage{graphicx}
\usepackage{epstopdf}
\usepackage{enumitem}
\usepackage{subcaption} 
\usepackage{geometry}
\usepackage{hyperref}
\usepackage{tikz}
\usepackage[all]{xy}
\usepackage{multirow}
\usepackage{anysize}
\usepackage{lineno}

\newtheorem{theorem}{Theorem}[section]
\newtheorem{definition}[theorem]{Definition}

\newtheorem{example}[theorem]{Example}
\newtheorem{proposition}[theorem]{Proposition}
\newtheorem{remark}[theorem]{Remark}
\newtheorem{corollary}[theorem]{Corollary}
\newtheorem{lemma}[theorem]{Lemma}

\newtheorem*{th1*}{Theorem 1}
\newtheorem*{th2*}{Theorem 2}
\newtheorem*{cor1*}{Corollary 1}
\newtheorem*{cor2*}{Corollary 2}
\newtheorem*{prediction*}{Prediction}
\newtheorem*{definition*}{Definition}
\newtheorem*{example*}{Example}
\newtheorem*{lemma*}{Lemma}
\newtheorem*{theorem*}{Theorem}

\newcommand{\beq}{\begin{equation}}
\newcommand{\eeq}{\end{equation}}

\newcommand{\FAT}[1]{\mbox{{$\mathbb{#1}$}}}
\newcommand{\NN}{\FAT{N}}

\newcommand{\RR}{\FAT{R}}

\DeclareMathOperator*{\argmax}{argmax}

\newcommand{\commentout}[1]{}

\def\wl{\par \vspace{\baselineskip}}

\begin{document}

\title{Cooperative Equilibrium beyond Social Dilemmas: \\ Pareto  Solvable Games}

\author{Valerio Capraro\footnote{Centrum Wiskunde \& Informatica (CWI), The Netherlands} \quad Maria Polukarov\footnote{University of Southampton, United Kingdom}  \quad Matteo Venanzi$^{\dag}$ \quad Nicholas R. Jennings$^{\dag}$
}

\maketitle

\begin{abstract}
A recently introduced concept of \emph{cooperative equilibrium}~\cite{capraro2013model}, based on the assumption that players have a natural attitude to cooperation, has been proven a powerful tool in predicting human behaviour in social dilemmas. In this paper, we extend this idea to more general game models, termed \emph{Pareto  solvable} games, which in particular include the Nash Bargaining Problem and the Ultimatum Game. We show that games in this class possess a unique pure cooperative equilibrium. Furthermore, for the Ultimatum Game, this notion appears to be strongly correlated with a suitably defined variant of the Dictator Game. We support this observation with the results of a behavioural experiment conducted using Amazon Mechanical Turk, which demonstrates that our approach allows for making  statistically precise predictions of average behaviour in such settings. 
\end{abstract}

\section{Introduction}\label{sec:introduction}
\noindent
Over recent decades, numerous experimental studies of human behaviour in economic scenarios have demonstrated the implausibility of traditional game-theoretic solutions based on payoff maximisation, rational choice and equilibria~\cite{camerer2003behavioral}. One of the most influential examples in experimental economics, offering evidence against such rationally self-interested play, is the Ultimatum Game~\cite{GSS82}.\footnote{Other examples include the Prisoner's Dilemma, the Public Goods Game, and the Dictator Game (see, e.g.,~\cite{camerer2003behavioral})} In this model, one player proposes a division of a given budget, $b$, between himself and the other player, who can either accept or reject the proposal. If the recipient accepts, the budget is split accordingly; if he rejects, both players get nothing. 

Clearly, if the recipient is a payoff maximiser, he would accept any non-zero offer, so the proposer should offer the smallest amount possible. However, as experimental results show (see, e.g.,~\cite{GSS82,kahneman1986fairness,camerer2003behavioral,larrick1997claiming,wallace2007heritability,zak2007oxytocin,roth1991bargaining,wells2013strategic}), human subjects tend to deviate from this equilibrium behaviour. Specifically, the recipients often reject low offers, and the proposers often propose higher shares. While there is a considerable quantitative variation across studies, the following regularities can be observed:
\begin{itemize}
	\item[I.] There are virtually no offers above $0.5b$, where $b$ denotes the given budget~\cite{fehr1999theory}.
	\item[II.] The vast majority of the offers lie in the interval $[0.3b,0.5b]$~\cite{rand2013evolution}.	
	\item[III.] There are almost no offers below $0.2b$~\cite{fehr1999theory,rand2013evolution}.
	\item[IV.] Offers below $0.25b$ are rejected with high probability~\cite{fehr2003nature}.
\end{itemize}

These results inspired several attempts to develop suitable models to explain human behaviour more accurately. Some were based on strategic arguments along the idea that higher offers are more likely to be accepted~\cite{roth1991bargaining,wells2013strategic}; others used the so-called  ``inequity aversion'' model, postulating that people are averse to unfair situations and so propose a fairer budget split~\cite{fehr1999theory,henrich2001search,lin2002using}. In this line, the classical Fehr-Schmidt inequity aversion model~\cite{fehr1999theory} is relatively successful in explaining the regularities I--IV. However, the parameter estimation method in this model (involving two parameters for each of the players---see Appendix for more details) has come under serious criticism in the recent literature~\cite{binmore2010experimental} due to the difficulty of generalising this method to a wider set of economic games. Furthermore, the recently proposed concept of ``benevolence'' (that is, making decisions towards increasing the opponent's payoff beyond one's own)~\cite{CSMN}, has challenged the very idea of inequity aversion (see also~\cite{charness2002understanding,engelmann2004inequality}). Alternatively, Rand et al.~\cite{rand2013evolution} used stochastic evolutionary game theory to show that natural selection favours fairness. However, while their model showed excellent results in matching population averages, there was a wider variance in individual strategies than is normally present in empirical data. In particular, their model yields a large number of individuals who make and demand offers above 50 per cent of the budget~\cite{fowler2013random}. 

Based on these findings, Camerer~\cite{camerer2003behavioral} suggested that human players in the Ultimatum Game are most likely to be motivated by a combination of strategic incentives and inequity aversion. In this work, we explore this idea and propose a solution concept that compounds these two factors. We note that, even though Fehr and Schmidt~\cite{fehr1999theory} also based their model on a similar approach, our work differs from theirs in a number of conceptual aspects, which in particular include the way we define the utility function. Specifically, in our model the strategic component comes in the form of a tendency to maximise a utility depending on the social welfare, while the inequity aversion component corresponds to a donation of the player in a certain variant of the Dictator Game~\cite{kahneman1986fairness} that we associate with the given instance of the Ultimatum Game. Both these behavioural tendencies have been reported by previous experimental studies~\cite{isaac1988group,goeree2001ten,nowak2006five,rand2013human,CJR,CSMN,fehr2003nature,burnham2003engineering,branas2006poverty,branas2007promoting,charness2008s,engel2011dictator,franzen2013external}. 

In more detail, our model builds on the concept of cooperative equilibrium, previously proposed for two-player one-shot social dilemmas~\cite{capraro2013model}, iterated two-player social dilemmas~\cite{CVPJ}, and $n$-player one-shot highly symmetric social dilemmas~\cite{barcelo2015group,capraro2015group}. We extend the definition to more general game classes, which we term here Pareto solvable games, due to their properties. These games in particular include the aforementioned Ultimatum Game, as well as other fundamental models, such as the Nash Bargaining Problem and the Coordination Game. We show that these games possess a unique cooperative equilibrium. 

Importantly, cooperative equilibrium in the Ultimatum Game gives a solution that explicitly predicts the regularities III and IV observed in previous experiments. Moreover, this solution determines a precise correlation between a player's behaviour in the Ultimatum Game and its corresponding variant of the Dictator Game. Using this correlation, we can give an a posteriori explanation for regularities I and II. 

Essentially, the idea behind the concept of cooperative equilibrium is to reduce the original game to another game, called the \emph{induced game}, which differs from the original game only in the set of allowed strategy profiles. It turns out that in the game induced by the Ultimatum Game, the proposer can not make proposals below some $x\%>0$ of the initial budget and the recipient can no longer reject. That is, the induced game is exactly the variant of the Dictator Game where the recipient is initially endowed with $x\%$ of the budget, and the proposer has to decide whether to make an additional donation to the recipient from the remaining part of the budget. The parameter $x$ represents the level of players' cooperation and their perception of money. Following similar techniques to those used  in~\cite{capraro2013model} to compute this parameter, we obtain $x=25\%$. In other words, cooperative equilibrium predicts that, on average, the proposer's behaviour in the Ultimatum Game is equivalent to the dictator's behaviour in the variant of the Dictator Game where the dictator is given $75\%$ of the original budget and the recipient has $25\%$ of the budget. That is, the offer in the Ultimatum Game should be equal to the donation in this variant of the Dictator Game, augmented by $25\%$ of the budget initially given to the recipient.

To test this prediction, we conducted a behavioural experiment via Amazon Mechanical Turk~\cite{paolacci2010running,horton2011online,rand2012promise}. We recruited US subjects, who were randomly assigned to two different treatments. In the first treatment, each participant was given a role of either the proposer or the recipient in the Ultimatum Game with a stake of  $\$0.40$. The proposers were then asked to play the following variant of the Dictator Game: they were given additional $\$0.30$ (i.e., 75\% of the budget in the Ultimatum Game) and paired with another anonymous MTurk worker who was given only $\$0.10$ (the remaining 25\%). In the second treatment, the games were played in the reversed order.

As our results demonstrate, cooperative equilibrium correctly predicts the majority of human behaviours in such settings. Specifically, we observed that about $80\%$ of our experimental subjects followed the cooperative equilibrium prediction. Interestingly, the remaining $20\%$ of the subjects who deviated from the prediction were those who behaved completely inconsistently in the two games: that is, they chose to split the budget equally in the Ultimatum Game, but kept everything to themselves in the Dictator Game. We note that this type of subject appears problematic for the inequity aversion models as well, and we believe, deserves a separate consideration (see Section~\ref{sec:discussion} for a thorough discussion).

Finally, this paper adds to the experimental evidence for a positive correlation between proposals in the Ultimatum Game and donations in the Dictator Game. To date, we are only aware of two published experimental studies where participants played both of these games, and one of them~\cite{peysakhovich2014humans} reported a significant positive correlation between the strategies in the Ultimatum Game and the Dictator Game.\footnote{In particular, \cite{exadaktylos2013experimental} did not directly report any correlation between the strategies in the two games. However, upon our request, one of the authors (Antonio Esp\'in) have confirmed that they also observe such a correlation} In this context, our work contributes to formally explain such a correlation between the two games through our newly defined concept of cooperative equilibrium for Pareto solvable games.

The paper unfolds as follows. We start with describing the games in consideration in Section~\ref{sec:preliminaries}. Then, in Section~\ref{sec:cooperative equilibrium}, we define our model. In detail, we first formalise our notion of fairness and then determine which strategies the players would consider feasible in this respect. We then proceed and introduce the class of Pareto  solvable games and formally define a cooperative equilibrium for these games. Our theoretical results are presented in Section~\ref{sec:results}, where we compute cooperative equilibria for the Ultimatum Game and the Nash Bargaining Problem. Specifically, we show that these games are indeed Pareto solvable and have a unique cooperative equilibrium, in which both players play pure strategies. Further, Section~\ref{sec:experiment} presents the experiment we conducted to test our predictions about the behaviour in the Ultimatum Game and its correlation with the behaviour in the Dictator Game. We conclude with a discussion of our results and future directions in Section~\ref{sec:discussion}. Additionally, the Appendix provides full instructions for our experimental subjects.


\section{Preliminaries}\label{sec:preliminaries}
\noindent
We start with the definitions of games considered in this paper. Our interest in these game classes is motivated by the two following reasons. First, they model situations of primary importance, such as negotiation and coordination problems. Second, they provide empirical evidence of behaviours, which are inconsistent with traditional game-theoretic solutions. In particular, the Ultimatum Game that we use as our main example, has long been in the centre of attention in experimental economics. 
\wl
\textbf{Ultimatum Game (UG).} In this game, the first player (the \emph{proposer}) is endowed with a sum of money (or, \emph{budget}), $b$, and needs to propose how to divide this money between himself and the other player (the \emph{recipient}). That is, he chooses the amount $p(b) \in [0, b]$ to propose to the recipient. The recipient gets to decide whether to accept or reject the proposal. If he accepts, the money is split according to the proposal: that is, the recipient gets $p(b)$, and the proposer keeps the remaining $b - p(b)$. If the recipient rejects, both players end up with no money at all. We denote this game by UG$(b)$.

\textbf{Dictator Game (DG).} In this game, similarly to the Ultimatum Game, the first player (the \emph{dictator}) determines how to split a budget $b$ between himself and the second player. The second player (the \emph{recipient}) simply accepts the donation made by the dictator. Thus, the recipient's role is entirely passive and has no effect on the outcome of the game. We denote this game by DG$(b)$.

For our study, it is also convenient to define a modified version of the Dictator Game, where the recipient may have some initial budget, $b_R$, and the dictator has to decide, how much (if at all) to donate to the recipient from his own budget, $b_D$, in addition to the recipient's initial budget. Thus, the dictator chooses a strategy $d(b_D, b_R) \in [0, b_D]$, and the recipient has no say. We denote this game by DG$(b_D,b_R)$.

\textbf{Nash Bargaining Problem (NBP).} This is a two-player game used to model bargaining interactions. In this game, two players demand a portion of a given budget, $b$. If the total amount requested by the players is less than $b$, they both get their request. If the total request exceeds $b$, both get nothing. 

\textbf{Coordination Game (CG).} In this game, each player has to select one of the two strategies, $A$ or $B$. Both players prefer any coordinated profile $(A,A)$ or $(B, B)$ to any uncoordinated profile, $(A,B)$ or $(B, A)$. However, the players' individual utilities from playing $(A,A)$ or $(B, B)$ may be different. 

Note also, that $(A,A)$ and $(B, B)$ are pure strategy Nash equilibria. Now, any probability distribution over these two equilibria is a correlated equilibrium, and the players have to decide which of the infinitely many possible such equilibria should be played. If they disagree and choose different distributions, they are likely to receive zero payoffs, while all bargaining solutions define correlated equilibria.
\wl


\section{Cooperative equilibrium for Pareto  solvable games}\label{sec:cooperative equilibrium}
\noindent
%
In this section, we extend the definition of cooperative equilibrium to a larger class of single-shot $n$-player games, which we term Pareto solvable. This class, in particular, contains the Ultimatum Game and the Nash Bargaining Problem (and hence, the Coordination Game) as above.

Our key idea is that players would divide into coalitions and act cooperatively as long as they believe that the resulting outcome of the game is ``fair'' to the members of each coalition. In Section \ref{subsec:undominated fair strategies}, we formalise our notion of fairness for a given coalition of players, and then determine which strategies the players would consider feasible  in this respect, given a particular coalition structure, that is, a partition of players into coalitions. We then proceed and introduce the class of Pareto  solvable games, and formally define a cooperative equilibrium for these games in Section \ref{subsec:cooperative equilibrium}.  


\subsection{Undominated fair strategies and games generated by coalition structures}\label{subsec:undominated fair strategies}
\noindent
Let $\mathcal G=\left (N,(S_i,\pi_i)_{i\in N}\right)$ be a normal-form game with a set $N$ of $n$ players, and, for all $i\in N$, a finite set of strategies $S_i$ and a monetary payoff function $\pi_i: S \rightarrow\RR$, where $S = \times_{j\in N}S_j$. For any given coalition (i.e., a subset) of players $C \subseteq N$, we write $S_C = \times_{j\in C}S_j$ for the set of strategy profiles restricted to $C$, and we use $-C$ to denote its complement, $N\setminus C$. Similarly, for a single player $i$, we use  $-i$ for the set $N\setminus\{i\}$ of all players but $i$. 

We denote by $\Delta(X)$ the set of probability distributions on a finite set $X$. Thus, $\Delta(S_i)$ defines the set of mixed strategies for player $i\in N$, and his expected payoff from a mixed strategy profile $\sigma$ is given by $\pi_i(\sigma) = \sum_{s\in S} \pi_i(s)\sigma_1(s_1)\cdot\ldots\cdot\sigma_{n}(s_{n})$, where $s=(s_1,\ldots,s_n)$.

Note that while the function $\pi_i$ defines a monetary payoff of player $i$ in game $\mathcal G$, different players may have different perceptions of the same amount of money they gain from the game. Formally, we express this idea using an \emph{aggravation} function Agg$_i(x,y)$, defined for $x\geq y$, which represents the aggravation that player $i$ experiences when renouncing a monetary payoff of $x$ and receiving a monetary payoff of $y$. We assume the following:
\begin{itemize}
\item Agg$_i(x,y)$ is nonnegative: 
\begin{itemize}
\item Agg$_i(x,y)\geq0$, for all $x\geq y$, and 
\item Agg$_i(x,y)=0$ if and only if $x=y$;
\end{itemize}
\item Agg$_i(x,y)$ is continuous;
\item Agg$_i(x,y)$ is strictly increasing in $x$ and strictly decreasing in $y$.
\end{itemize}
\begin{example}\label{ex:every stake}
{\rm An intuitive example of an aggravation function is given by {\rm Agg}$_i(x,y)=u_i(x)-u_i(y)$, where $u_i(k)$ stands for the utility of player $i$ from getting a monetary payoff of $k$. It is often assumed (see~\cite{kahneman1979prospect}, p. 278) that $u_i$ is concave for positive payoffs (gains) and convex for negative payoffs (losses). That is, the same increase in the gains is perceived more significant for small amounts and less significant as the gains grow, and the other way around for the losses. In low-stake experiments, however, this effect is likely to be negligible, so one may find the simple aggravation function {\rm Agg}$_i(x,y)=x-y$ practically useful, since it is parameter-free and easy to compute.}
\end{example}
Now, let $p=(p_1,\ldots,p_k)$ denote a partition of the player set $N$, defining a possible coalition structure that the players can form in $\mathcal G$. We denote the set of all possible partitions of $N$ by $\mathcal{P}(N)$. Each element $p_\alpha$ of this partition is interpreted as a coalition: that is, the players in the same $p_\alpha$ play together. It is natural to expect that the players in a coalition $p_\alpha$ would only accept those outcomes, which pay the members of the coalition in a fair way (at the very least). With the aggravation functions as above, we define a notion of fairness as follows. 

Let $p_{\alpha}$ be a coalition in the coalition structure $p$ and $s$ be a profile of pure strategies. Let $W_{p_\alpha}(s)=\sum_{i\in p_\alpha}\pi_i(s)$ denote the social welfare of coalition  $p_{\alpha}$, and set:
\begin{align}
M_{p_\alpha}(s)=\max_{ i\in p_\alpha}{\rm Agg}_i(W_{p_\alpha}(s),\pi_i(s)),
\end{align}
\begin{align}
m_{p_\alpha}(s)=\min_{ i\in p_\alpha}{\rm Agg}_i(W_{p_\alpha}(s),\pi_i(s)).
\end{align}
That is, $M_{p_\alpha}(s)$ (resp., $m_{p_\alpha}(s)$) denotes the maximum (resp., the minimum) aggravation that a player in coalition $p_{\alpha}$ has from getting his share of payoff, $\pi_i(s)$, compared to the total coalitional payoff, $W_{p_\alpha}(s)$. Intuitively, the smaller the difference between $M_{p_\alpha}(s)$ and $m_{p_\alpha}(s)$, the more fair the distribution of payoffs among the coalition members. Hence, we make the following definition. 
\begin{definition}\label{def:fair strategies}
{\rm A profile of pure strategies $s$ is said to be \emph{fair} for a given coalition $p_\alpha \in p$, if it minimises the difference 
\begin{align}
M_{p_\alpha}(s)-m_{p_\alpha}(s).
\end{align}}
\end{definition}
Obviously, there may be multiple fair profiles of strategies, resulting in different individual payoffs to the members of a coalition (cf. Example~\ref{ex:generated game}). We order the set of fair strategies using the super-dominance relation, introduced independently in~\cite{capraro2013solution} and~\cite{HaPa13}. 
\begin{definition}\label{def:super-dominance}
{\rm A profile of pure strategies $s_{p_\alpha}'\in S_{p_\alpha}$ for coalition $p_\alpha$ \emph{super-dominates} a profile $s_{p_\alpha}\in S_{p_\alpha}$, if for all $s_{-p_\alpha}, s_{-p_\alpha}'\in S_{-p_\alpha}$ one has $\pi_i(s'_{p_\alpha},s_{-p_\alpha}')> \pi_i(s_{p_\alpha},s_{-p_\alpha})$, $\forall i\in p_\alpha$.}
\end{definition}
In words, a strategy is super-dominated by another strategy if the largest payoff that can be obtained by playing the former strategy is strictly smaller than the smallest payoff that can be obtained by playing the latter.

Now, let $\text{Fair}(\mathcal G,p_{\alpha})$ be the set of pure strategies that are fair for coalition $p_\alpha$ and let $\text{Max}(\text{Fair}(\mathcal G,p_{\alpha}))$ contain its maximal elements w.r.t. the partial order induced by the super-dominance relation\footnote{Since strategy sets are finite, $\text{Fair}(\mathcal G,p_{\alpha})$ always has a maximal element. Moreover, for the Ultimatum Game and the Bargaining problem, we will show that $\text{Fair}(\mathcal G,p_{\alpha})$ is always a singleton, thus implying $\text{Max}(\text{Fair}(\mathcal G,p_{\alpha})) = \text{Fair}(\mathcal G,p_{\alpha})$. For these games the assumption of the strategy sets' finiteness can then be removed and, in fact, for technical convenience, we will assume these sets are continuous.}. These will be named \emph{undominated fair strategies} for coalition $p_\alpha$.

Given this, we next say that a coalition structure $p$ generates a new game, $\mathcal G_p$, where players within each coalition $p_\alpha \in p$ cooperate to maximise their total payoff, playing fair, undominated strategies. Formally, 
\begin{definition}\label{def:generated game}
{\rm Given a game $\mathcal G$ and a coalition structure $p$ over $N$, let $\mathcal G_p$ be the game \emph{generated by} $p$ as follows. The player set in $\mathcal G_p$ corresponds to the set of coalitions $p_\alpha \in p$, for each of which its set of mixed strategies is given by the closure of the convex hull of $\text{Max}(\text{Fair}(\mathcal G,p_\alpha))$ and the payoffs are defined as the sums of payoffs of the members of $p_\alpha$.}
\end{definition}
Let $p_s$ denote what we call the (fully) \emph{selfish} coalition structure where each player acts individually---i.e., $p_s=(\{1\},\ldots,\{n\})$. Since every strategy is fair in such a structure, the game $\mathcal G_p$ is then obtained from $\mathcal G$ by simply deleting all super-dominated strategies. In general, however, $\mathcal G_p$ may appear very different from $\mathcal G$ and, importantly, it need \emph{not} be the game generated by the strategies that maximise the total welfare. In particular, this may happen for the (fully) \emph{cooperative} coalition structure $p_c =(N)$ where all the players act collectively, forming the \emph{grand} coalition, $N=\{1,\ldots,n\}$. See the following Example~\ref{ex:generated game} for illustration.  
\begin{example}\label{ex:generated game}
{\rm Let $\mathcal G$ be the following variant of the Prisoner's dilemma:
$$
  \begin{array}{ccc}
     & \text{C} & \text{D} \\
    \text{C} & 10,10 & 1,20 \\
    \text{D} & 20,1 & 2,2 \\
  \end{array}
$$
Assume that players have the same perception of money, which we denote {\rm Agg}. Since {\rm Agg} is strictly decreasing in the variable $y$, one can easily verify that for the  cooperative coalition structure $p_c=(\{1,2\})$, the only fair strategies are $(C,C)$ and $(D,D)$. Indeed,
\begin{align*}
M_{\{1,2\}}(C,D)-m_{\{1,2\}}(C,D) = M_{\{1,2\}}(D,C)-m_{\{1,2\}}(D,C) = \text{ \rm Agg}(21,1)-\text{\rm Agg}(21,20) > 0, 
\end{align*}
while 
\begin{align*}
& M_{\{1,2\}}(C,C)-m_{\{1,2\}}(C,C)  \, \, = \text{ \rm Agg}(20,10)-\text{\rm Agg}(20,10)   \\
 = \quad  & M_{\{1,2\}}(D,D)-m_{\{1,2\}}(D,D)  = \text{ \rm Agg}(4,2)-\text{\rm Agg}(4,2)  = 0.
\end{align*}
Now, note that $(D,D)$ is super-dominated by $(C,C)$, which is then the only maximal fair strategy, and hence, $\mathcal G_{p_c}$ is a one-player, one-strategy game. Yet, the only strategy (i.e., $(C,C)$) admissible in this game does \emph{not} maximise the social welfare. 
}
\end{example}


\subsection{Pareto solvable games and cooperative equilibrium}\label{subsec:cooperative equilibrium}
\noindent
Having defined the games generated by coalition structures, we now use them to formalise our extended notion of cooperative equilibrium. For simplicity of exposition, we restrict the definition to the case of two-player games, which in particular contain the Ultimatum Game and the Bargaining problem, which are the focus of this paper. We later discuss how this definition generalises to games with $n$ players. 

Consider a two-player game $\mathcal G$ and fix a coalition structure $p$. For the game $\mathcal G_p$ generated by $p$, let $\text{Nash}(\mathcal G_p)$ denote the set of its mixed strategy Nash equilibria. We then say that player $i$  \emph{plays according} to the coalition structure $p$, if he plays his corresponding marginal strategy $\sigma_i$ of some Nash equilibrium profile $\sigma \in \text{Nash}(\mathcal G_p)$, given that other players play $\sigma_{-i}$. 
We say that player $i$  \emph{deviates} from $p$ if he applies a unilateral profitable deviation from an equilibrium profile $\sigma \in \text{Nash}(\mathcal G_p)$---that is, he plays a strategy $\sigma_i'$ against $\sigma_{-i}$, for some $\sigma \in \text{Nash}(\mathcal G_p)$, such that $\pi_{i}(\sigma_i',\sigma_{-i})>\pi_i(\sigma)$. Note that while no player can deviate from the selfish coalition structure $p_s = (\{1\}, \{2\})$ by definition, deviations may be applied from the cooperative structure $p_c =(\{1, 2\})$. 

Given a coalition structure $p$, a player $i$ and a subset $D\subseteq N \setminus \{i\}$ of other players, let $\tau^i_{D}(p)\in[0,1]$ denote the subjective probability that player $i$ assigns to the event that the set of players $D$ deviates from $p$ for the sake of their individual interests, knowing that player $i$ plays according to the coalition structure $p$. In a two-player game, we only have $D=-i$ or $D=\emptyset$, and so $\tau^i_{-i}(p) = 1 - \tau^i_{\emptyset}(p)$. In~\cite{capraro2013model}, the probabilities $\tau^i_{D}(p)$ were computed directly using equation (\ref{eq:parameters}) below. However, in light of the more recent developments, where these probabilities have been used to show the consistency of the cooperative equilibrium theory with the Social Heuristics Hypothesis~\cite{R13a} or predicting human behaviour in iterated social dilemmas~\cite{CVPJ}, here we prefer to treat them as parameters, so that the following properties hold: 
\begin{itemize}
\item[(A1)] If player $-i$ has no profitable deviation from any $\sigma \in \text{Nash}(\mathcal G_p)$, then $\tau^i_{-i}(p)=0$; 
\item[(A2)] If for all $\sigma\in\text{Nash}(\mathcal G_p)$ and for all deviations $\sigma_{-i}'$ of player $-i$ from $\sigma$ there is a strategy $\sigma_{i}'$ such that $\pi_{-i}(\sigma_i',\sigma_{-i}')<\pi_{-i}(\sigma)$, then $\tau^i_{D}(p) < 1$.
\end{itemize}
That is, if there is no incentive for a player to deviate from any equilibrium profile under a given coalition structure, then the probability of deviating from the structure is zero; and, if deviating from a coalition structure may lead to an eventual loss, then the probability of deviating is less than 1.  In particular, (A1) implies that $\tau^i_{-i}(p_s)=0$, since the sets of Nash equilibria of $\mathcal G$ and $\mathcal G_{p_s}$ coincide, and hence no unilateral deviations are profitable.

Now, let $e^i_{\emptyset}(p)$ be the minimal payoff that player $i$ can get if nobody abandons the coalition structure $p$ (i.e., $e^i_{\emptyset}(p) = \min_{\sigma \in \text{Nash}(\mathcal G_p)} \pi_i(\sigma)$), and let $e^i_{-i}(p)$ be the minimal payoff of player $i$ when he plays according to some Nash equilibrium of $\mathcal G_p$ and $-i$ deviates from that equilibrium profile. We define the \emph{value} of a coalition structure $p$ for player $i$ as an expected value of these minimal payoffs, that is,  
\begin{align}
v_i(p) = e^i_{\emptyset}(p)\tau^i_{\emptyset}(p)+e^i_{-i}(p)\tau^i_{-i}(p).
\end{align} 
In what follows, we consider a game $\mathcal G$ as a tuple $\left (N,(S_i,\pi_i)_{i\in N},\tau^i_{D}(p)_{i\in N, D\subseteq N\setminus\{i\},p \in \mathcal{P}(N)}\right)$, where $\left (N,(S_i,\pi_i)_{i\in N}\right)$ corresponds to the original strategic game form and $\tau^i_{D}(p) \in [0,1]$ are given parameters satisfying the properties (A1) and (A2).

With this notation, we now define the class of \emph{Pareto  solvable} games, as follows.
\begin{definition}
{\rm A two-player game $\mathcal G$ is called \emph{Pareto  solvable} if there exists a coalition structure $p^*$ that maximises the values $v_1(p)$ and $v_2(p)$ for both players at once.}
\end{definition}
The class of Pareto  solvable games is quite large. Indeed, under the assumption of the deviation probability parameters $\tau^i_{D}(p)$ being symmetric (that is, $\tau^1_{2}(p)=\tau^2_{1}(p)$ for $p=p_s, p_c$), which is implied, in particular, by the formulas for their computation proposed in~\cite{capraro2013model} and~\cite{CVPJ}, it clearly contains all major two-player social dilemmas, such as the Prisoner's dilemma, the Traveler's dilemma and the Bertrand competition. Moreover, regardless of the exact values of $\tau^i_{D}(p)$'s, as long as they satisfy the properties (A1) and (A2), it also includes the Ultimatum Game and the Nash Bargaining Problem (and hence, the Coordination game), as we demostrate in Section~\ref{sec:results} below.

For these games, we now define the notion of cooperative equilibrium. We use the values $v_i(p^*)$ to represent the players' beliefs about the payoffs they expect to obtain in the game. In other words, these values make a tacit binding among the players, such that they only play those strategies that garantee them payoffs of at least $v_i(p^*)$. 

\begin{definition}
{\rm Let $\mathcal G$ be a Pareto  solvable game and $p^*$ a coalition structure that maximises the values $v_i(p)$ for both players. The \emph{induced} game ${\rm Ind}(\mathcal G, p^*)$ is the restriction of $\mathcal G$ where the set of (mixed) strategies is reduced to include only those profiles $\sigma \in \times_{i \in N} \Delta(S_i)$, for which $\pi_i(\sigma)\geq v_i(p^*)$ for all $i \in N$.}
\end{definition}

Finally, since the set of strategy profiles in the induced game is non-empty, convex and compact, it possesses a Nash equilibrium, and we can make the following definition.
\begin{definition}\label{defin:cooperative equilibrium}
{\rm A \emph{cooperative equilibrium} of a Pareto  solvable game $\mathcal G$ is defined as a Nash equilibrium of the induced game ${\rm Ind}(\mathcal G, p^*)$.}
\end{definition}

\begin{remark}
{\rm For the sake of simplicity, we have defined the cooperative equilibrium only for two-player Pareto solvable games. The extension to $n$-player games is straightforward. For example, the value for player $i$ of a coalition structure $p$ will be defined as 
$$
v_i(p) = \sum_{C\subseteq N\setminus\{i\}}e_C^i(p)\tau_C^i(p),
$$
where $\tau_C^i(p)$ is the subjective probability that player $i$ assigns to the event that players in $C$ abandon the coalition structure $p$ and $e_C^i(p)$ is the minimum payoff that player $i$ would get if he plays according to the coalition structure $p$ and players in $C$ abandon it. Given this definition, one may define $n$-player Pareto solvable games as those games for which there is a coalition structure $p^*$ such that $v_i(p)$ is maximised, for all $i=1,\ldots,n$. In a similar fashion, all other definitions extend to $n$ player games. }
\end{remark}

\section{Cooperative equilibria of Ultimatum Games and Nash Bargaining Problems}\label{sec:results}
\noindent
In this section, we compute cooperative equilibria for some important subclasses of Pareto  solvable games: namely, for the Ultimatum Game and the Nash Bargaining Problem. Specifically, we demonstrate that these games are indeed Pareto solvable and have a unique cooperative equilibrium in pure strategies. 


\subsection{Purifiable games}\label{subsec:purifiable games}
\noindent
The existence of a pure strategy cooperative equilibrium is due to the fact that the induced games of UG and NBP are what we term \emph{purifiable}, in the sense that for every mixed strategy $\sigma_i$ of any player $i$, there exists a pure strategy $s_i$, which is a profitable deviation from $\sigma_i$, whichever strategy $\sigma_{-i}$ is played by the other player(s). Formally,   
\begin{definition}\label{defin:purifiable games}
{\rm A game $\mathcal G$ is called \emph{purifiable} if for each player $i \in N$ and any mixed strategy $\sigma_i\in\Delta (S_i)$, there is a pure strategy $s_i \in S_i$ such that $\pi_i(s_i,\sigma_{-i})\geq \pi_i(\sigma_i,\sigma_{-i})$, for all $\sigma_{-i}\in \times_{j \neq i}\Delta(S_j)$.}
\end{definition}
For such games, it then suffices to consider only pure strategies. Now, in both the UG and the NBP models, the two players share a unit of payoff (i.e., we normalise the budget $b$ to be equal to 1) between them -- that is, the total sum of the players' payoffs is always 1. As the following Lemma~\ref{lem:unique solution} shows, there is then only one way to share the payoffs so that the players have equal aggravations.
\begin{lemma}\label{lem:unique solution}
Let {\rm Agg}$_i(x,y)$ be the aggravation function of player $i = 1,2$ in game $\mathcal G$. Then, the equation ${\rm Agg}_1(1,1-x) = {\rm Agg}_2(1,x)$, with $x\in[0,1]$, has a unique solution $x^*$. Moreover, $x^* \neq 0,1$.
\end{lemma}
\begin{proof}
The proof is straightforward. Let $f_1(x)= {\rm Agg}_1(1,1-x)$ and $f_2(x)= {\rm Agg}_2(1,x)$. Then,
\begin{itemize}
\item $f_1(x)$ is a strictly increasing continuous function with $f_1(0)=0$ and $f_1(1)>0$; and 
\item $f_2(x)$ is a strictly decreasing continuous function with $f_2(0)>0$ and $f_2(1)=0$. 
\end{itemize}
Hence, there is a unique  point $x^*$ where $f_1(x^*)=f_2(x^*)$, and $x^*\neq 0,1$. 
\end{proof}
We now use Lemma~\ref{lem:unique solution} to describe cooperative equilibria of the Ultimatum Game and the Nash Bargaining Problem, as follows. 


\subsection{Ultimatum Games}\label{subsec:UG}
\noindent

\begin{theorem}\label{thm:UG}
Let $\mathcal G$ be the Ultimatum Game with players P (the proposer) and R (the recipient), and budget $b$. Let $x^*$ be the solution to ${\rm Agg}_P(1,1-x) = {\rm Agg}_R(1,x)$. Then,
\begin{itemize}
\item $\mathcal G$ is Pareto  solvable, where the values for both players are maximised by the cooperative coalition structure $p_c$ and equal:
$$
v_P(p_c) = b\left(1 - x^*\right)\qquad\text{and}\qquad v_R(p_c) = bx^* \tau^R_{\emptyset}(p_c). 
$$
\item The induced game ${\rm Ind}(\mathcal G,p_c)$ is purifiable, and its purification coincides with the variant of the Dictator Game with parameters $b_D = b\left(1-x^* \tau^R_{\emptyset}(p_c)\right)$ and $b_R =  bx^* \tau^R_{\emptyset}(p_c)$, where the proposer plays the role of a dictator and donates to the recipient the amount of at most $bx^*\tau^R_{P}(p_c)$.
\end{itemize} 
\end{theorem}
\begin{proof}
First, observe that the game generated by the selfish coalition structure $p_s = \left(\{P\}, \{R\}\right)$ coincides with the original game (i.e., $\mathcal G_{p_s}=\mathcal G$), as there are no super-dominated strategies. Now, while this game may have different Nash equilibria, the lowest payoff for both players is obtained at the equilibrium profile $(0,R)$ where the proposer makes an offer of 0 to the recipient, and the recipient rejects. In this case, both players get nothing, and since by (A1) the probability of deviation from $p_s$ is 0 for both players, they both have a zero value for the selfish coalition structure $p_s$. Thus, $v_P(p_s)=v_R(p_s)=0$.

For the cooperative coalition structure, $p_c = \left(\{P, R\}\right)$, we now demonstrate that its generated game $\mathcal G_{p_c}$ only allows a single strategy profile $(bx^*, A)$, where the proposer makes an offer of $bx^*$ and the recipient accepts. To this end, we only need to show that, for all $x\in[0,1]$ such that $x \neq x^*$, one has
\begin{align}\label{eq:inequality}
M_{\{P, R\}}(bx^*, A)-m_{\{P, R\}}(bx^*, A)<M_{\{P, R\}}(bx, A)-m_{\{P, R\}}(bx, A).
\end{align}
But this is straightforward from the definition of $x^*$ as the unique solution for 
\begin{align*}
M_{\{P, R\}}(bx, A)-m_{\{P, R\}}(bx, A) = b\left|{\rm Agg}_P(1,1-x)-{\rm Agg}_R(1,x)\right| = 0.
\end{align*}
implying that the left-hand side of inequality~(\ref{eq:inequality}) equals zero, and the right-hand side is strictly positive. 

Now, if the proposer is assumed to make an offer of $bx^* > 0$, then it is not beneficial for the recipient to deviate from accepting the offer, and hence, by (A1), $\tau^P_{\emptyset}(p_c)=1$. Thereby, since $e^P_{\emptyset}(p_c)=b(1-x^*)$, we have $v_P(p_c)=b(1-x^*)>0$. 

On the other hand, if it is believed that the recipient will accept, then the proposer has an incentive to deviate and make a smaller offer. In the worst case for the recipient, the proposer offers him zero, so  $e^R_{P}(p_c)=0$. We thus get $v_R(p_c)=bx^* \tau^R_{\emptyset}(p_c)$. Finally, by (A2), we have that $\tau^R_{\emptyset}(p_c) >0$, which, coupled with $bx^* > 0$, yields $v_R(p_c)>0$. 

Hence, the cooperative coalition structure $p_c$ maximises the values for both players, so the Ultimatum Game is Pareto solvable. 

We now turn to prove the second part of the theorem, stating that the induced game ${\rm Ind}(\mathcal G,p_c)$ is purifiable. To avoid unnecessary technical complications, we only consider finitely supported strategies for the proposer, even if his strategy set is given by the continuous interval $[0,b]$. That is, we assume that any proposer's mixed strategy $\sigma_P$ is given by a finite sequence of coefficients $\left(\lambda_{(1)},\ldots,\lambda_{(K)}\right)$, for some $K \in \NN$, such that $\lambda_{(k)}>0$ represents the probability of playing a pure strategy $s_{(k)} \in [0,b]$. Hence, $\sum_{k=1}^K \lambda_{(k)}=1$, and w.l.o.g. we assume that the strategies $s_{(k)}$ are ordered in ascending order. 

We next demonstrate that any such mixed strategy $\sigma_P$ can be replaced by the pure strategy $s_{(1)}$. Indeed, for any mixed strategy $\sigma_R = \left(\lambda,1-\lambda\right)$ of the recipient, with $\lambda$ being the probability of $s_R=A$ (accept), we get
\begin{align*}
\pi_P(\sigma_P, \sigma_R)
&= \lambda\left(\lambda_{(1)}(b-s_{(1)})+\ldots+\lambda_{(K)}(b-s_{(K)})\right)\\
&< \lambda(b-s_{(1)})(\lambda_{(1)}+\ldots+\lambda_{(K)})\\
&= \lambda(b-s_{(1)}) = \pi_P(s_{(1)},\sigma_R).
\end{align*}
Similarly, any mixed strategy $\sigma_R = \left(\lambda,1-\lambda\right)$ of the recipient, where $\lambda < 1$, can be replaced by the pure strategy $s_R=A$, as for any strategy $\sigma_P$ of the proposer we have
\begin{align*}
\pi_R(\sigma_P, \sigma_R)
&= \lambda\left(\lambda_{(1)}s_{(1)}+\ldots+\lambda_{(K)}s_{(K)}\right)\\
&< \lambda_{(1)}s_{(1)}+\ldots+\lambda_{(K)}s_{(K)}\\
&= \pi_R(\sigma_R, A).
\end{align*}
Finally, in the purification of the induced game $\text{Ind}(\mathcal G,p_c)$, the proposer makes an offer of at least $v_R(p_c) = bx^* \tau^R_{\emptyset}(p_c)$, but of at most $b - v_P(p_c) = bx^*$, and the recipient must accept. The game then coincides with the variant of the Dictator Game with parameters $b_P = b\left(1-x^* \tau^R_{\emptyset}(p_c)\right) \geq b\left(1 - x^*\right)$ and $b_R =  bx^* \tau^R_{\emptyset}(p_c)$, where the donation that the dictator can make to the recipient does not exceed 
\begin{align*}
b_P - v_P(p_c) 
& = b\left(1-x^* \tau^R_{\emptyset}(p_c) - (1 - x^*)\right) \\
& = bx^* ( 1 - \tau^R_{\emptyset}(p_c)) = bx^*\tau^R_{P}(p_c).
\end{align*}
This completes the proof. 
\end{proof}
Theorem~\ref{thm:UG} implies that the Ultimatum Game has a unique pure strategy cooperative equilibrium, in which the proposer makes an offer of $x^* \tau^R_{\emptyset}(p_c)$, and the recipient accepts. Moreover, since the induced game of UG essentially coincides with a corresponding variant of the Dictator Game, the cooperative equilibrium predicts a correlation between these two games. Specifically, the theorem suggests that
\begin{corollary}\label{UGvsDG}
 In a Ultimatum Game $\mathcal G$, the proposer makes an offer of 
\begin{align}\label{eq:prediction with altruism}
bx^* \tau^R_{\emptyset}(p_c)+d\left(b\left(1 - x^* \tau^R_{\emptyset}(p_c)\right),bx^* \tau^R_{\emptyset}(p_c)\right),
\end{align}
where $d\left(b\left(1 - x^* \tau^R_{\emptyset}(p_c)\right),bx^* \tau^R_{\emptyset}(p_c)\right)$ is the donation that the proposer would make to the recipient if he were to act as the dictator in DG$(b_P, b_R)$  with parameters $b_P = b\left(1-x^* \tau^R_{\emptyset}(p_c)\right)$ and $b_R =  bx^* \tau^R_{\emptyset}(p_c)$, and we expect that $d\left(b\left(1 - x^* \tau^R_{\emptyset}(p_c)\right),bx^* \tau^R_{\emptyset}(p_c)\right) \leq bx^*\tau^R_{P}(p_c)$. 
\end{corollary}
Now, in order to obtain testable predictions of players' behaviour, we need to evaluate the parameters $x^*$ and $\tau^i_{D}(p)$. As it is often in experiments with anonymous participants, the average experimental subjects may overgeneralise their experience and expect their opponents to have the same perception of money as they do, so, in order to make predictions of average behaviour in anonymous games, we assume that all players have the same aggravation function. That is, for any two subjects $i$ and $j$, we have ${\rm Agg}_i(x,y)  = {\rm Agg}_j(x,y)$. Applying this assumption in Lemma~\ref{lem:unique solution}, we get $x^*=0.5$, which allows us to make the following individual level predictions:
\begin{itemize}
\item \emph{Local prediction.} In a UG experiment with anonymous players, \emph{most} proposers will offer to their opponents a payment of $b \cdot 0.5  \tau^R_{\emptyset}(p_c)+d\left(b\left(1 - 0.5  \tau^R_{\emptyset}(p_c)\right), b \cdot 0.5 \tau^R_{\emptyset}(p_c)\right)$.
\end{itemize} 
To estimate the probabilities of deviation from a given coalition structure, we adopt the idea from~\cite{capraro2013model}. Specifically, these probabilities depend on the \emph{incentive} for each player to abandon a coalition structure, and his possible \emph{loss} from doing so. Here, we define the incentive, $I_{-i}(p)$, for a player $-i \in N$ to abandon the coalition structure $p$, as the aggravation in renouncing to the maximal possible gain that $-i$ can obtain by leaving the coalition structure $p$, while his opponent (player $i$) plays according to $p$:  
$$
I_{-i}(p) = \sup\left\{\text{Agg}_{-i}\left(\pi_{-i}(\sigma^p_{i}, \sigma_{-i}), \pi_{-i}(\sigma^p_{i}, \sigma_{-i}^p)\right) : \sigma^p\in\text{Nash}(\mathcal G _p), \sigma_{-i}\in\Delta (S_{-i})\right\}
$$
However, player $i$ may also deviate from $p$ (whether he anticipates player $-i$'s deviation or not), thus yielding a loss for player $-i$. We define such a loss, $L_{-i}(p)$, of a player $-i$ when he leaves a coalition structure $p$ in order to maximise his payoff as the maximal aggravation that he may experience, while player $i$ also may deviate from the coalition structure $p$ (either following his selfish interests or anticipating player $-i$'s deviation). Formally, 

{\scriptsize $$
L_{-i}(p) = \sup\left\{\text{Agg}_{-i}\left(\pi_{-i}(\sigma_{i}, \sigma_{-i}), \pi_{-i}(\sigma^p_{i}, \sigma^p_{-i})\right) : \begin{array}{l} \; \sigma^p\in\text{Nash}(\mathcal G _p), \sigma_{-i} \in \argmax_{\Delta (S_{-i})} \pi_{-i}(\sigma^p_{i}, \sigma_{-i}), \sigma_{i} \in \Delta (S_{i}) \\ \textrm{ s.t.} \;\;\;\, \pi_{i}(\sigma_i, \sigma^p_{-i}) > \pi_{i}(\sigma^p_i, \sigma^p_{-i}) \;\;\; \textrm{or} \;\;\;  \pi_{i}(\sigma_i, \sigma_{-i}) > \pi_{i}(\sigma^p_i, \sigma_{-i})\end{array}\right\}
$$}

\noindent That is, for each $\sigma^p\in\text{Nash}(\mathcal G _p)$, mixed strategy $\sigma_{-i}$ maximises the payoff of player $-i$ from $(\sigma^p_i,\sigma_{-i})$, and $\sigma_{i}$ runs over the set of  $i$-deviations from either $(\sigma_i^p,\sigma_{-i}^p)$ or $(\sigma^p_i,\sigma_{-i})$. 

Given this, we finally set: 
\begin{align}\label{eq:parameters}
\tau^i_{-i}(p)=\frac{I_{-i}(p)}{I_{-i}(p)+L_{-i}(p)}
\end{align}

As we previously mentioned in Example~\ref{ex:every stake}, in a low-stake experiment, function Agg$_i(x,y)=x-y$ appears a good approximation of the real players' aggravation. For such a function, it is immediate that $I_{-i}(p_c)=x^*=0.5$ and $L_{-i}(p_c)=1-x^*=0.5$. We therefore can make the following prediction: 
\begin{itemize}
\item \emph{Global prediction.}  In a UG experiment with anonymous players and low stakes, it is expected that the average offer is composed of $25\%$ of the total endowment, $b$, and the average donation in DG$(0.75b, 0.25b)$.
\end{itemize} 
We test these predictions experimentally in Section~\ref{sec:experiment}. We conclude the section with presenting analogous results for the Nash Bargaining Problem.

\subsection{Nash Bargaining Problem}\label{subsec:NBP}
\noindent
Consider the classical Nash Bargaining Problem, where two players bargain on how to share $b$ units of surplus, where they get their claims satisfied if and only if these claims are feasible, that is, they sum up to less than or equal to $b$. Then,
\begin{theorem}\label{thm:NBP}
In the only pure cooperative equilibrium of the Nash Bargaining Problem, one player asks for $bx^*$ and the other player claims $b\left(1-x^*\right)$.
\end{theorem}
The proof of Theorem~\ref{thm:NBP} follows similar lines to the proof of Theorem~\ref{thm:UG}, and hence is omitted. Finally, since for players with the same perception of money we have $x^*=0.5$, this yields the following corollary.
\begin{corollary}
If the players have the same perception of money, then in the unique cooperative equilibrium of the Nash Bargaining Problem, each player asks for exactly half of the endowment.
\end{corollary}


\section{Behavioural experiment investigating the correlation between  Ultimatum Game and Dictator Game}\label{sec:experiment}
\noindent
In this section, we report on the experiment we conducted to test the correlation between the UG offers and the DG donations as our model predicts. We start with the outline of the experiment, and then discuss the results we obtained.  

\subsection{Design}\label{subsec:design}
\noindent
We recruited US subjects (a.k.a. \emph{workers}) via Amazon Mechanical Turk (AMT)~\cite{paolacci2010running,horton2011online,rand2012promise}. As is common in experiments with human participants (either online or in physical laboratories), AMT workers receive a baseline payment for participation and can also earn an additional bonus that depends on how they perform in the game. The main advantage of the AMT experiments is in that they  are easily implementable and relatively inexpensive, compared to laboratory experiments. Nevertheless, it has been shown that data gathered using AMT agree both qualitatively and quantitatively with data collected in physical labs~\cite{horton2011online,rand2012promise,suri2011cooperation}. 

However, when planning such an experiment, one should be aware of the fact that AMT workers may try to play the same game multiple times in order to get a larger bonus and/or select random strategies in order to minimise the time spent on completing the task. To control for these issues, we used Qualtrics---a survey builder that allowed us to exclude those workers who have already participated in the experiment (by checking their AMT ID and their IP address) or those who failed to answer the comprehension questions that we designed for each part of the task to make sure that the workers understand the task and would make sensible decisions. 

The task that the workers had to perform in the experiment consisted of two parts: play a given role in the Ultimatum Game and in the Dictator Game. The experiment contained two treatments, differing in the order in which the UG and the DG were played. In both treatments, subjects earned $\$0.20$ for participating and received a bonus that was determined from their performance in the two games. In more detail: 
\begin{itemize}
\item \emph{Treatment 1.}  Subjects were randomly assigned a role of either the proposer/dictator or the recipient. First, they played the UG as follows. The proposers were given a budget of $\$0.40$ and asked how much, if at all, they would offer to an anonymous recipient they have been randomly paired with, while the recipient may either accept or reject the offer. Responders, at the same time, were asked to declare their minimal acceptable offer. \\ 
Second, the DG was played. The proposers (playing now the role of a dictator) were notified that they have been assigned another anonymous worker to interact with. They were given another $\$0.30$ and informed that their current partner is endowed with only $\$0.10$. Then they were asked how much, if at all, they would donate to the other participant, who has no say and can only accept their offer.
\item \emph{Treatment 2.} As above, but playing the DG first and the UG second. 
\end{itemize} 

In both games, we restricted the set of possible strategies to multiples of 2 cents, within the given budget. Thus, a proposer in the Ultimatum Game could only make an offer of $\{0, 2, 4, \ldots, 40\}$ cents. Full instructions for participants are presented in the Appendix.


\subsection{Results}\label{subsec:results}
\noindent
The results for Treatment 1 are summarised in Table~\ref{table:stat study 1} and Table~\ref{table:prediction study 1}. We had 172 subjects that played as proposers in the Ultimatum Game and as dictators in the Dictator Game.

In detail, the mean offer in the UG was $18.18$c, which constitutes $45\%$ of the budget. More than half of the subjects (about $65\%$) offered the equal split, which was also the modal and the median offer. Only one subject acted consistently with the subgame perfect equilibrium for payoff-maximising preferences, which is to offer the recipient nothing at all. We note that these results appear consistent with those reported in other studies~\cite{kagel1995handbook,camerer2003behavioral,oosterbeek2004cultural}. 

In the DG, the mean offer was $5.40$c. The distribution of donations was essentially bimodal: $42\%$ of the subjects decided to act selfishly and make no donation at all, while $40\%$ of the participants acted so as to minimise inequity and donate $10$c. These results are qualitatively similar to those presented in previous studies, although in previous reports the selfish behaviour was observed significantly more often than the inequity averse behaviour~\cite{engel2011dictator}. This difference may appear related to the framing effect caused by the fact that the DG was played after the UG in this treatment. However, we cannot reject the null hypothesis that donations in Treatment 1 and Treatment 2 (differing by the order of the UG and DG) are actually the same. In fact, we ran the non-parametric Wilcoxon Rank-sum test, i.e., the standard hypothesis test for data samples that cannot assumed to be normally distributed like our case  \cite{wilcoxon1945individual}, and the returned significance value is $P=0.1969$, which does not reject the null hypothesis. 

\begin{figure}
  \begin{subfigure}[b]{0.5\textwidth}
    \includegraphics[width=\textwidth]{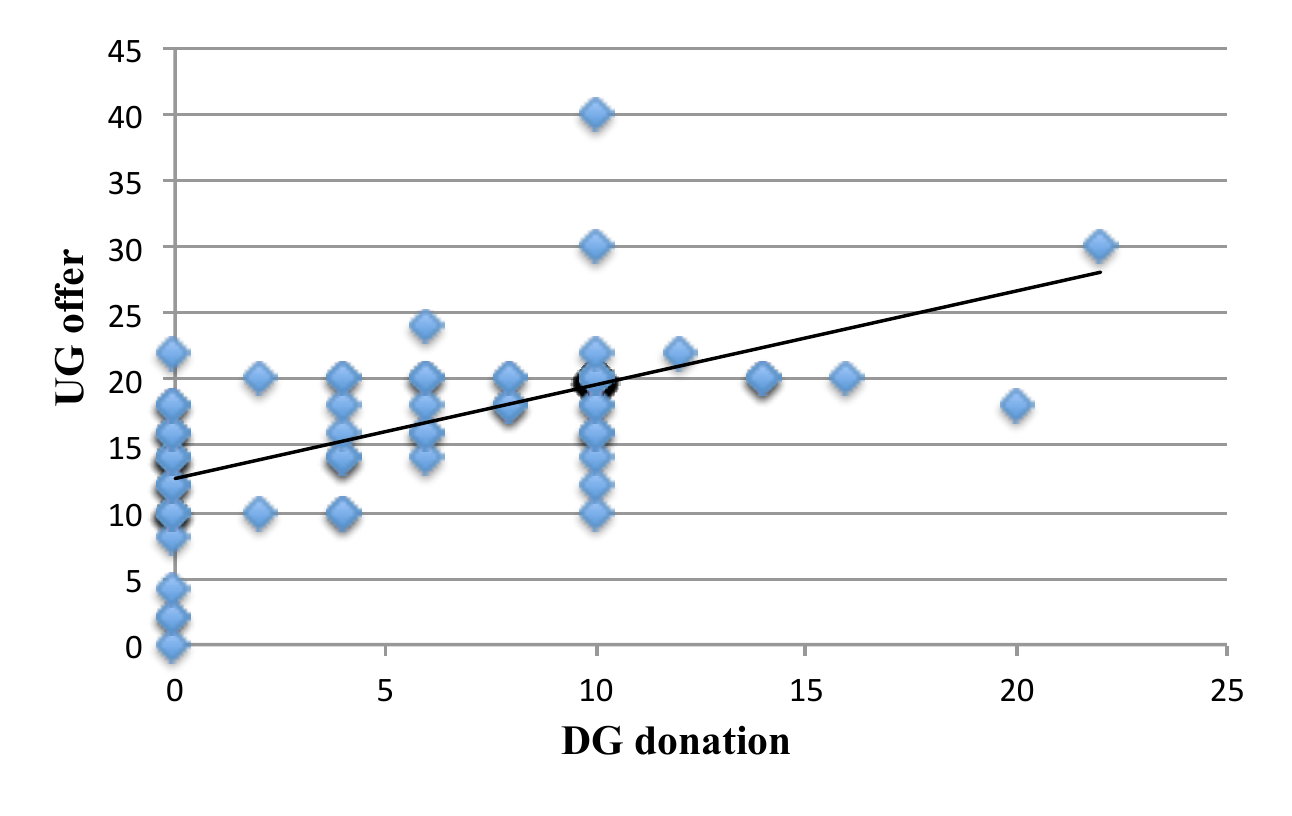}
    \caption{Treatment 1: UG before DG}
    \label{fig:tr1}
  \end{subfigure}
  \begin{subfigure}[b]{0.5\textwidth}
    \includegraphics[width=\textwidth]{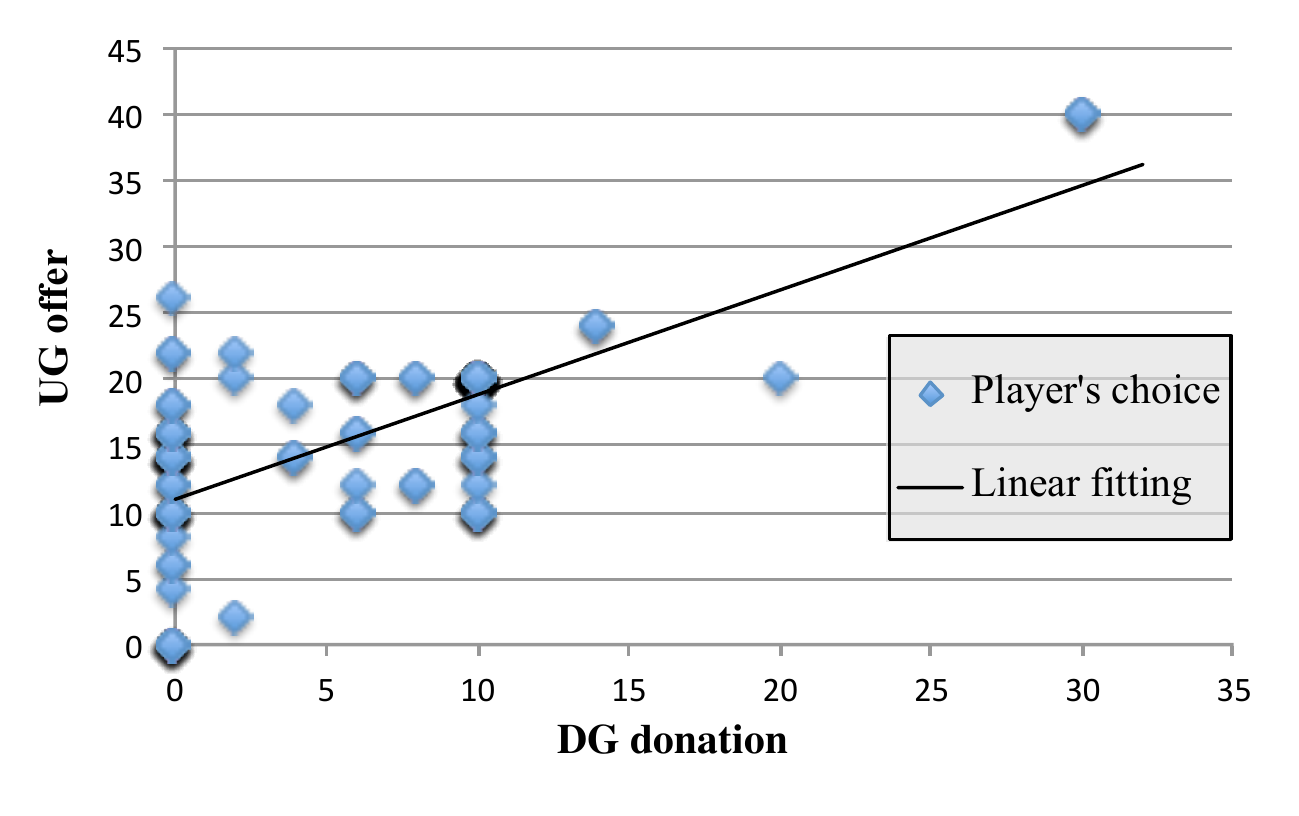}
    \caption{Treatment 2: DG before UG}
    \label{fig:tr2}
  \end{subfigure}
  \caption{Scatter plot of the  players' choices in the two treatments: UG played after DG (a) and DG played after UG (b), after filtering the players with $0c$ donation in DG and an equal split offer in UG.}
\end{figure}

In contrast, we find a significant positive correlation between the UG offers and the DG donations (Spearman's $\rho$ rank correlation test: $\rho = 0.3412, P<.0001$). In particular, we observe that $43\%$ of the subjects played exactly according to our global prediction based on the cooperative equilibrium:
\begin{align}\label{eq:prediction}
p(40) = d(30, 10) + 10
\end{align}
That is, their offer in the Ultimatum Game with $b=40$c was composed of their donation in the Dictator Game with parameters $b_D= 0.75b=30$ and $b_R =0.25b=10$, and additional $10$c (i.e.,  $25\%$ of the total endowment in the UG).

We test the quality of this prediction by fitting the following linear equation:
\begin{align}\label{eq:pol-fitting}
p(40) = c_{\text{fitting}}(1)d(30, 10) + c_{\text{fitting}}(0)
\end{align}
where the coefficients $c_{\text{fitting}}(0)$, $c_{\text{fitting}}(1)$ are estimated through least square error. This linear fit is also showed in Figure 1a. We compare the coefficients with the coefficients $c_{\text{model}}(0)$ and $c_{\text{model}}(1)$ fitted by our prediction model by computing their corresponding mean squared errors, MSE$_{\text{fitting}}$ and MSE$_{\text{model}}$. We find that $c_{\text{fitting}}(0)=16.3912$ and $c_{\text{fitting}}(1)=0.3320$, with the mean squared error MSE$_{\text{fitting}}$ = $17.6383$, while the cooperative equilibrium model prescribes $c_{\text{model}}(0)=10$ and $c_{\text{model}}(1)=1$, resulting in the mean squared error MSE$_{\text{model}}$ = $37.0467$.

At first glance, our prediction may seem quite inaccurate, compared to the linear fitting. However, it turns out that the difference in equations~(\ref{eq:prediction}) and~(\ref{eq:pol-fitting}) is mainly due to a subset of 42 subjects (about $24\%$)  who demonstrated completely inconsistent behaviour in playing the two games---that is, they decided for the equal split in the UG (offering the recipient $20$c), but acted selfishly in the DG (donating $0$c). Note that this type of subjects is also very problematic for the inequity aversion models, as they act in an inequity averse way in the UG and the opposite in the DG. We refer to Section~\ref{sec:discussion} for a more detailed discussion about this issue.

Taking out these 42 inconsistent subjects allows us to observe significantly stronger correlations. When only consistent participants were taken into account, the mean UG offer was equal $17.60$c and the mean DG donation was $7.15$c. The correlation between the UG offers and the DG donation also gets much stronger (Spearman's $\rho$ rank correlation: $\rho=0.7136, P<.0001$). Moreover, $57\%$ of these subjects played exactly according to equation~(\ref{eq:prediction})). Finally, the linear fitting now gives $p_{\text{fitting}}(0)=12.5308$, $p_{\text{fitting}}(1)=0.7086$ and MSE$_{\text{fitting}}=14.6281$, which is much closer to the cooperative equilibrium model with $p_{\text{model}}(0)=10$, $p_{\text{model}}(1)=1$ and MSE$_{\text{model}}=16.7077$. This trend is also shown by the scatter plot of the players' choices together with their linear fitting in Figure \ref{fig:tr1}.

\begin{table}[ht]
\centering
\begin{tabular}{|c | c | c | c | c | c | c | c | c | c }
\hline\hline
& Mean Offer & Mean donation & Correlation \\ [0.5ex]
\hline
All & 18.18 &5.40&0.3412$^{***}$ \\
No (20,0) & 17.60&7.15 & 0.7136$^{***}$\\[1ex]
\hline
\end{tabular}
\vspace{0.3cm}
\caption{Statistical analysis of Treatment 1. Spearman's $\rho$ rank correlation test. $^{***}$ stands for $P<.0001$.}\label{table:stat study 1}
\end{table}

\begin{table}[ht]
\centering
\begin{tabular}{|c | c | c | c | c | c | c | c | c | c }
\hline\hline
& Linear fitting & Model & MSE fitting & MSE model \\ [0.5ex]
\hline
All &$p(0)=16.40, p(1)=0.33$&$p(0)=10,p(1)=1$&17.64 &37.05\\
No (20,0) & $p(0)=12.53, p(1)=0.71$&$p(0)=10,p(1)=1$ & 14.63&16.71\\[1ex]
\hline
\end{tabular}
\vspace{0.3cm}
\caption{Comparison between polynomial fitting model at degree $n=1$ and predictions of the cooperative equilibrium in Treatment 1.}\label{table:prediction study 1}
\end{table}

The results for Treatment 2 are given in Tables~\ref{table:stat study 2} and~\ref{table:prediction study 2}. We had 198 subjects that played as dictators in the Dictator Game and as proposers in the Ultimatum Game. The mean donation was $5.39$c, while the distribution of donations was essentially bimodal with half of the participants donating $0$c and $37\%$ making the inequity minimising donation of $10$c. While these distributions of donations seems to be more left-weighted than the one we observed in Treatment 1, as we mentioned before, we cannot exclude the null hypothesis that they are actually the same (Wilcoxon rank-sum: $P=0.1969$). The mean UG offer was $16.77$c, that is $42\%$ of the budget. Slightly more than half of the subjects ($53\%$) chose the equal split and ten subjects acted consistently with the subgame perfect equilibrium for payoff-maximising preferences, which is to offer $0$c or $2$c.

Similarly to Treatment 1, we find a significant positive correlation between the UG offers and the DG donations (Spearman's $\rho$ rank correlation: test: $\rho = 0.3322, P<.0001$). In this case, $38\%$ of the subjects played exactly according to our global prediction stated in Equation~(\ref{eq:prediction}). Again, we tested the quality of a linear fitting (see Equation~(\ref{eq:pol-fitting})). We obtained $c_{\text{fitting}}(0)=14.2689$ and $c_{\text{fitting}}(1)=0.5079$ and the mean squared error  MSE$_{\text{fitting}}$ = $35.5953$, while the cooperative equilibrium prediction gives  $c_{\text{model}}(0)=0.10, c_{\text{model}}(1)=1$ and the mean squared error  MSE$_{\text{model}}$ = $40.7273$.

Also in this case, the difference is mainly caused by the inconsistent behaviour of 42 subjects ($21\%$) who donated $0$c in the DG and proposed the equal split in the UG. When only consistent participants were considered, the mean UG offer was equal to $15.91$c and the mean DG donation was $6.27$c. The correlation between the UG offers and the DG donation also gets much stronger (Spearman's $\rho$ rank correlation: $\rho=0.6377, P<.0001$). This correlation is also apparent in the scatter plot and linear fitting in Figure~\ref{fig:tr2}. Moreover, $49\%$ of these subjects played exactly according to Equation~(\ref{eq:prediction}), and the polynomial fitting gave  $c_{\text{fitting}}(0)=10.9531$, $c_{\text{fitting}}(1)=0.7907$ and MSE$_{\text{fitting}}=23.1418$, compared with the cooperative equilibrium prediction with $c_{\text{model}}(0)=10$, $c_{\text{model}}(1)=1$ and MSE$_{\text{model}}=24.7695$.

\begin{table}[ht]
\centering
\begin{tabular}{|c | c | c | c | c | c | c | c | c | c }
\hline\hline
& Mean donation & Mean offer & Correlation \\ [0.5ex]
\hline
All & 5.39 &16.77&0.3322$^{***}$ \\
No (0,20) & 6.27&15.91 & 0.6377$^{***}$\\[1ex]
\hline
\end{tabular}
\vspace{0.3cm}
\caption{Statistical analysis of Treatment 2. Spearman's $\rho$ rank correlation test. $^{***}$ stands for $P<.0001$.}\label{table:stat study 2}
\end{table}

\begin{table}[ht]
\centering
\begin{tabular}{|c | c | c | c | c | c | c | c | c | c }
\hline\hline
& Linear fitting & Model & MSE fitting & MSE model \\ [0.5ex]
\hline
All &$p(0)=14.27, p(1)=0.51$&$p(0)=10,p(1)=1$&35.60 &40.73\\
No (0,20) & $p(0)=10.95, p(1)=0.79$&$p(0)=10,p(1)=1$ & 23.14&24.77\\[1ex]
\hline
\end{tabular}
\vspace{0.3cm}
\caption{Comparison between polynomial fitting model at degree $n=1$ and predictions of the cooperative equilibrium in Treatment 2.}\label{table:prediction study 2}
\end{table}


\section{Discussion}\label{sec:discussion}
\noindent
In this paper, we have extended the concept of cooperative equilibrium from two-player social dilemmas to a more general class of games containing, as important instances, the Ultimatum Game  and the Nash Bargaining Problem. We used this concept to obtain theoretical predictions of human players' behaviour in these games. In particular, we derived a correlation between the offers that proposers make to recipients in the UG and the donations they would make in the corresponding variant of the Dictator Game where the proposer is endowed with $75\%$ of the total budget, while the recipient has the remaining $25\%$ upfront. This correlation is particularly interesting because it can be tested in the laboratory.

Our model explicitly predicts the regularities III and IV outlined in Section~\ref{sec:introduction}. Specifically, it suggests that on average, UG proposers would offer at least 25\% of their budget to recipients, which is consistent with the experimental findings showing that there are almost no offers below $0.2b$ (regularity III). Also notice, that accepting less than 25\% of the proposer's budget is excluded from the set of feasible strategy profiles in the game induced by the cooperative coalition structure, which is consistent with regularity IV, stating that offers below $0.25b$ get rejected with high probability. Finally, the correlation we established between the UG offers and the DG donations gives an \emph{a posteriori} explanation of regularities I and II: indeed, virtually every donation in the DG is bounded above by the equal split~\cite{engel2011dictator}.

In order to test the predicted correlation between the offers in the UG and the donations in the DG, we conducted an experiment via Amazon Mechanical Turk and found that about $80\%$ of experimental subjects played according to our prediction. As for the remaining 20\%, we must point out that these subjects seem to form a special class showing inconsistent behaviour in the two games: these subjects made a zero donation in the Dictator Game, but offered half of the budget to the recipient in the Ultimatum Game. Such behaviour appears problematic not only for our approach, but also for the standard inequity aversion model introduced by Fehr and Schmidt~\cite{fehr1999theory}. This model assumes that player $i$'s utility from the allocation $(x_i,x_{-i})$ of monetary payoffs between himself and player $-i$, is given by:
\begin{align}
u_i(x_i,x_{-i})=x_i-\alpha_i\max\{x_{-i}-x_i, 0\}-\beta_i\max\{x_i-x_{-i}, 0\}
\end{align}
where $0\leq\beta_i\leq\alpha_i$ are individual parameters, representing the extent to which player $i$ regrets being worse off (parameter $\alpha_i$) or enjoys being better off (parameter $\beta_i$) than the other player. By estimating the joint distribution $(\alpha,\beta)$ of the parameters, the model gives 30\% of players with $(\alpha,\beta)=(0,0)$, 30\% of players with $(\alpha,\beta)=(0.5,0.25)$, 30\% with $(\alpha,\beta)=(1,0.6)$, and 10\% with $(\alpha,\beta)=(4,0.6)$. However, this distribution is not consistent with the aforementioned subjects choosing equal split of the budget in the UG, but a zero donation in the DG. To see this, first observe that if a subject makes a zero donation in the Dictator Game, then his parameter $\beta$ must be smaller than $0.5$. Indeed, a subject with $\beta=0.5$ is indifferent to any donation and a subject with $\beta > 0.5$ would donate half of the endowment. Also note that a person donating zero in the DG but splitting the budget equally in the UG cannot be of the type with $(\alpha,\beta)=(0.5,0.25)$, since then he would donate a half of his endowment also in the DG. It follows that, according to the Fehr and Schmidt's estimation of the parameters, the only possibility is that such a person has the type with $(\alpha,\beta)=(0,0)$. We further make the following few observations (see the Appendix for the full proof):
\begin{enumerate}
\item Subjects of type $(\alpha,\beta)=(0,0)$ who believe they are playing with subjects of type $(\alpha,\beta)=(0,0)$, would donate nothing in the DG and offer nothing in the UG.
\item Subjects of type $(\alpha,\beta)=(0,0)$ who believe they are playing with subjects of type $(\alpha,\beta)=(0.5,0.25)$, would donate nothing in the DG and offer 10c in the UG.
\item Subjects of type $(\alpha,\beta)=(0,0)$ who believe they are playing with subjects of type $(\alpha,\beta)=(2,0.6)$, would donate nothing in the DG and offer 16c in the UG.
\item Subjects of type $(\alpha,\beta)=(0,0)$ who believe they are playing with subjects of type $(\alpha,\beta)=(4,0.6)$, would donate nothing in the DG and 17.78c in the UG.
\end{enumerate}
Now, it is clear from the above that Fehr and Schmidt's model does not account for the players offering equal split in the UG and making a zero donation in the DG. The question is though who are these subjects and why do they act in such an apparently inconsistent way? We note that, although their behaviour deviates from our global prediction that is based on an explicit computation of the model parameters, yet it is consistent with Theorem~\ref{thm:UG}, corresponding to the case where $ \tau^R_{\emptyset}(p_c)=1$ and both the proposer and the recipient have the same perception of money. Recalling that  $\tau^R_{\emptyset}(p_c)$ represents the probability that the recipient assigns to the event that the proposer does not deviate from the cooperative coalition structure, it follows that the cooperative equilibrium concept explains such subjects as more cooperative than the average. Testing this (as we see it, non-trivial) hypothesis is an interesting direction for future research.

There is also another interesting novel prediction provided by our cooperative equilibrium model. Consider the situation where it is common knowledge that the proposer is much richer than the recipient. This could be formalised by saying that the aggravation function of the proposer, for fixed $y$, increases much slower than that of the recipient. Consequently, the intersection point $x^*$ found in Lemma~\ref{lem:unique solution} lies much closer to 1. Keeping $\tau_{R,P}(p_c)$ constant, this implies that the prediction of the cooperative equilibrium is that, in this situation, the proposer's offer would be higher than the one he would make if they were equally rich. Consequently, the cooperative equilibrium predicts that proposers offer more to poor players than to equally rich players. Such a behaviour has been qualitatively observed in the Dictator Game~\cite{branas2007promoting}, but we are not aware of any study investigating it in the Ultimatum Game. In future research, it would be interesting to test this prediction experimentally.



\bibliographystyle{plain} 
\bibliography{valerio}



\section*{Appendix}

\subsection*{Technical details}
\noindent
According to Fehr and Schmidt's model, the utility for player $i$ when player $i$ receives $x_i$ and player $-i$ receives $x_{-i}$ is:
\begin{align}
u_i(x_i,x_{-i})=x_i-\alpha_i\max(x_{-i}-x_i,0)-\beta_i\max(x_i-x_{-i},0),
\end{align}
where $0\leq\beta_i\leq\alpha_i$ are individual parameters, representing the extent to which player $i$ regrets to be worse off (parameter $\alpha_i$) and better off (parameter $\beta_i$) than the other player. We say that a player has type $(\alpha,\beta)$ if he plays according to Fehr and Schmidt's utility function with parameters $\alpha$ and $\beta$.

\begin{proposition}
Consider the Ultimatum Game where the proposer has a budget of 40c and the Dictator Game where the dictator has a budget of 30c and the recipient already owns 10c. The following statements hold:
\begin{enumerate}
\item If the proposer has type $(\alpha,\beta)=(0,0)$ and believes he is playing with a recipient of type $(\alpha,\beta)=(0,0)$, than he would donate nothing in the DG and offer nothing in the UG.
\item  If the proposer has type $(\alpha,\beta)=(0,0)$ and believes he is playing with a recipient of type $(\alpha,\beta)=(0.5,0.25)$, than he would donate nothing in the DG and offer 10c in the UG.
\item  If the proposer has type $(\alpha,\beta)=(0,0)$ and believes he is playing with a recipient of type $(\alpha,\beta)=(2,0.6)$, than he would donate nothing in the DG and offer 16c in the UG.
\item  If the proposer has type $(\alpha,\beta)=(0,0)$ and believes he is playing with a recipient of type $(\alpha,\beta)=(4,0.6)$, than he would donate nothing in the DG and offer 17.78c in the UG.
\end{enumerate}
\end{proposition}
\begin{proof}
The first statement is trivial. As for the second one, the part concerning the DG is trivial. As for the part concerning the UG, we first observe that since the proposer has type $(0,0)$ then he makes the minimum offer such that the recipient accepts. Now, observe that rejecting an offer gives utility $0$ to the recipient, while accepting an offer gives utility negative only if that offer $x$ is smaller than 10c, in which case the recipient gets a utility of $2x-20 < 0$. Thus, the proposer offers 10c and this completes the proof of the second statement. The proof of the third and fourth statements are very much alike that of the second statement.
\end{proof}


\subsection*{Experimental instructions}
\noindent
Here we report full experimental instructions of Treatment 1. Instructions of Treatment 2 were exactly the same, apart from the order.
\wl
\textbf{\underline{Screen 1.}}

This is the first part of the HIT. You have been paired with another, anonymous MTurk worker. Your bonus will depend on your and his choices.

We give you a budget of 40 cents. Decide how much, if any, to offer to the other participant. At the same time, he will decide his minimal acceptable offer. 

If your offer is larger than or equal to his minimal acceptable offer, then he will get your offer and you will keep the rest. Otherwise, if your offer is smaller than his minimal acceptable offer, both of you get nothing.

The other person is real and will really make a decision.
\wl
\textbf{\underline{Screen 2.}}

Now we will ask you some comprehension questions to make sure you understand the game. You MUST answer all questions CORRECTLY to receive the payment. In case you fail one of these questions, you will immediately see the message "We thank you for your time spent taking this survey. Your response has been recorded", and you will not receive any payment.
\wl
\textbf{\underline{Screen 3.}}

Assume YOU offer 12 cents and the minimal acceptable offer of the OTHER PARTICIPANT is 10 cents. What is your bonus? (Here participants could choose between 0c, 10c, 12c, and 28c).

\wl
\textbf{\underline{Screen 4.}}

Assume YOU offer 12 cents and the minimal acceptable offer of the OTHER PARTICIPANT is 10 cents. What is his bonus? (Here participants could choose between 0c, 10c, 12c, and 28c).

\wl
\textbf{\underline{Screen 5.}}

Assume YOU offer 8 cents and the minimal acceptable offer of the OTHER PARTICIPANT is 14 cents. What is your bonus? (Here participants could choose between 0c, 8c, 14c, and 32c).

\wl
\textbf{\underline{Screen 6.}}

Assume YOU offer 8 cents and the minimal acceptable offer of the OTHER PARTICIPANT is 14 cents. What is his bonus? (Here participants could choose between 0c, 8c, 14c, and 32c).

\wl
\textbf{\underline{Screen 7.}}

Congratulations! You have passed the comprehension questions. Now, it's time to make your decision.

WHAT IS YOUR OFFER? (Here participants could select every \emph{even} amount of money between 0c and 40).

\wl
\textbf{\underline{Screen 8.}}

This is the second part of the HIT. Now, you have been paired with ANOTHER anonymous MTurker. This time your bonus depends ONLY on your own choice: the other participant CANNOT influence your bonus.

You are given a new budget of 30 cents and the other participant is given 10 cents. You have to decide how much, if at all, to DONATE to the other participant. The other participant has no say and will accept any donation.

Note that the decision is unilateral: nobody will have the choice to make a donation to you.

The other person is real and will really get your donation.

\wl
\textbf{\underline{Screen 9.}}

Now we will ask you some comprehension questions to make sure you understand the game. You MUST answer all questions CORRECTLY to receive the payment. In case you fail one of these questions, you will immediately see the message "We thank you for your time spent taking this survey. Your response has been recorded", and you will not receive any payment.

\wl
\textbf{\underline{Screen 10.}}

If you donate 10 cents, what is YOUR bonus? (Here participants could choose between 30c, 20c, 10c, and 0c).

\wl
\textbf{\underline{Screen 11.}}

If you donate 10 cents, what is the OTHER PARCIPANT's  bonus?  (Here participants could choose between 30c, 20c, 10c, and 0c).

\wl
\textbf{\underline{Screen 12.}}

If you donate 0 cents, what is YOUR bonus?  (Here participants could choose between 30c, 20c, 10c, and 0c).

\wl
\textbf{\underline{Screen 13.}}

If you donate 0 cents, what is THE OTHER PARTICIPANT's bonus?  (Here participants could choose between 30c, 20c, 10c, and 0c).

\wl
\textbf{\underline{Screen 14.}}

If you donate 12 cents, what is YOUR bonus? (Here participants chould choose between 22c, 18,c, 8c, and 12c).

 \wl
\textbf{\underline{Screen 15.}}

If you donate 12 cents, what is the OTHER PARTICIPANT's bonus? (Here participants chould choose between 22c, 18,c, 8c, and 12c).

 \wl
\textbf{\underline{Screen 16.}}

Congratulations! You have passed the comprehension questions. Now, it's time to make your decision.

WHAT IS YOUR DONATION?

\end{document}